\documentclass[acmsmall,nonacm]{acmart}
\usepackage{algorithm}
\usepackage[noend]{algpseudocode}
\usepackage{booktabs}
\usepackage{mathrsfs}
\usepackage{soul}

\makeatletter
\algrenewcommand\ALG@beginalgorithmic{\small}
\algrenewcommand\alglinenumber[1]{\scriptsize #1:}
\makeatother

\usepackage{adjustbox}

\algdef{SE}[SUBALG]{Indent}{EndIndent}{}{\algorithmicend\ }
\algtext*{Indent}
\algtext*{EndIndent}
\makeatletter
\newcommand{\multiline}[1]{
  \begin{tabularx}{\dimexpr\linewidth-\ALG@thistlm}[t]{@{}X@{}}
    #1
  \end{tabularx}
}
\makeatother

\newif\ifcomment
\commenttrue    

\newcommand{\junk}[1]{}
\newcommand{\suchthat}
{\;\ifnum\currentgrouptype=16 \middle\fi|\;}

\algnewcommand{\LineComment}[1]{\State \color{blue} \(\triangleright\) #1 \color{black}}

\algnewcommand{\BComment}[1]{\color{blue}  \Comment{#1} \color{black}}

\newcommand{\synchAlgo}{\mathcal{A}_{sync}}
\newcommand{\NonsynchAlgo}{\mathcal{A}_{semi}}
\newcommand{\synchronizer}{$\delta$}
\newcommand{\synchGraph}{\mathcal{G}_{sync}}
\newcommand{\synchGraphSnap}{H}
\newcommand{\NonsynchGraph}{\mathcal{G}_{semi}}

\newcommand{\mc}[1]{}

\title{Synchronization in Anonymous Networks Under Arbitrary Dynamics}
\author{Rida Bazzi}
\orcid{https://orcid.org/0000-0003-1872-1634}
\email{bazzi@asu.edu}
\author{Cameron Bickley}
\email{cpbickle@asu.edu}
\authornote{(NSF (CCF-2312537, CCF-2106917) and U.S.\ ARO (MURI W911NF-19-1-0233).}
\author{Anya Chaturvedi}
\orcid{https://orcid.org/0000-0002-4071-0467}
\email{anya.chaturvedi@asu.edu}
\authornotemark[1]
\author{Andr\'ea W. Richa}
\orcid{https://orcid.org/0000-0003-3592-3756}
\email{aricha@asu.edu}
\authornotemark[1]
\author{Peter Vargas}
\orcid{https://orcid.org/0009-0006-0669-819X}
\email{pjvargas@asu.edu}
\authornotemark[1]
\affiliation{\institution{School of Computing and Augmented Intelligence, Arizona State University}
\city{Tempe, AZ}
\country{USA}}


\begin{document}

\begin{abstract}
We present the $\delta$-Synchronizer, which works in non-synchronous dynamic networks under minimal assumptions. Our model allows for arbitrary topological changes without any guarantee of eventual global or partial stabilization and assumes that nodes are anonymous. This deterministic synchronizer is the first that enables nodes to simulate a dynamic network synchronous algorithm for executions in a semi-synchronous dynamic environment under a weakly-fair node activation scheduler, despite the absence of a global clock, node ids, persistent connectivity or any assumptions about the edge dynamics (in both the synchronous and semi-synchronous environments). 
We make the following contributions: (1) we extend the definition of synchronizers to networks with arbitrary edge dynamics; (2) we present the first synchronizer from the semi-synchronous to the synchronous model in such networks;
and  (3) we present non-trivial applications of the proposed synchronizer to existing algorithms. We assume an extension of the {\sc Pull} communication model by adding a single 1-bit multi-writer atomic register at each edge-port of a node. We show that this extension is needed and that synchronization in our setting is not possible without it. 
The $\delta$-Synchronizer operates with a multiplicative memory overhead at the nodes that is asymptotically logarithmic on the runtime of the underlying synchronous algorithm being simulated---in particular, it is logarithmic for polynomial-time synchronous algorithms. 
\end{abstract}

\maketitle

\section{Introduction}
\label{sec:intro}
Modern distributed systems, such as wireless sensor networks, mobile peer-to-peer systems, and biologically inspired swarm networks, often exhibit constantly changing and unpredictable communication dynamics that make achieving coordinated behavior between agents challenging, even for simple tasks. 
Achieving coordinated behavior in such a setting can be further compounded by asynchrony and agent anonymity (no identifiers). 
To simplify the design of distributed algorithms, researchers developed {\em synchronizers} that can transform algorithms designed under strong synchrony assumptions into algorithms that work correctly under weaker assumptions~\cite{awerbuch1985-complexity,shabtay1994-lowcomplex,ghaffari2023-nearoptimal}. 
Previous work on synchronizers considers systems in which agents have unique identifiers and the network is either static or is dynamic for some time but eventually stabilizes. In contrast, this paper 
considers a system of anonymous agents that communicate through a network with {\em arbitrary} {\em dynamics}, i.e., with {\em no restrictions on topological changes}, including {\em no assumptions on
eventual (local or global) stabilization}. 

In such networks, three different synchrony models are considered~\cite{Flocchini_Prencipe_Santoro_2019}. In the {\em synchronous} model, time is divided into stages and all nodes are active in every stage. 
 In the {\em semi-synchronous model} (see, e.g.,~\cite{flocchini2020-semisync}), time is also divided into stages, but some nodes might not be active in a given stage. In both models, an active node executes
 one action---which involves communication with their neighbors and a bounded amount of computation---per stage, and topology changes occur only at the beginning of a stage. In other words, the {\em semi-synchronous} mode is one where time is synchronous but nodes are activated asynchronously. In the {\em asynchronous} model, there is no synchronization of time nor node activations, so actions can take an arbitrary bounded amount of time to execute, and topological changes and node activations can happen at arbitrary times.

In this paper, we introduce our \emph{Arbitrary-Dynamics Synchronizer}, or {\em \synchronizer -Synchronizer} for short. The \synchronizer -Synchronizer is the first {\em deterministic} synchronizer that allows algorithms designed for a synchronous anonymous network with arbitrary adversarial dynamics to execute correctly in a semi-synchronous anonymous, arbitrary dynamics network under a {\em weakly-fair node activation scheduler}.
Specifically,  {\em our \synchronizer-Synchronizer transforms any algorithm $\synchAlgo$ designed for a 
{\em synchronous} dynamic network under arbitrary edge-dynamics given by a time-varying graph $\synchGraph$ into an algorithm $\NonsynchAlgo$
that correctly simulates $\synchAlgo$ under a {\em weakly-fair 
scheduler} in a {\em semi-synchronous} dynamic network under arbitrary edge-dynamics given by a time-varying graph $\NonsynchGraph$}.

Unlike other synchronizers that assume eventual stabilization and for which one can compare executions of the original algorithm and the transformed algorithm on the same stabilized network, in our setting, there is no guarantee that the network ever stabilizes and there is no guarantee that $\synchGraph$ and $\NonsynchGraph$ are identical, which should not be surprising. This necessitates a {\em non-triviality} requirement on synchronizers for networks with arbitrary dynamics. 
The overall requirements for such synchronizers are captured by the following three general conditions (in bold), 
which we then indicate how they are satisfied by our \synchronizer -Synchronizer. 
\begin{itemize}
        \item {\bf (Correctness) The simulated synchronous execution is valid}: We show that for any $\NonsynchGraph$ and $i\geq 0$, there exists a $\synchGraph$ such that the state of each node at the end of phase $i$ (where phases are maintained by the synchronizer) of the semi-synchronous execution of $\NonsynchAlgo$ under $\NonsynchGraph$ and a weakly-fair scheduler is equal to the state of each node at the end of $i$-th step of a synchronous execution of $\synchAlgo$ under $\synchGraph$. ({Theorem~\ref{thm:safety}})
         
         \item 
         {\bf (Non-triviality) Every possible outcome of a synchronous execution can be simulated}:
          We show
          that  
         for any $\synchGraph$ and $i\geq 0$, there exists a $\NonsynchGraph$ such that the state of each node at the end of step $i$ of the synchronous execution of $\synchAlgo$ under $\synchGraph$, is equal to the state of each node at the end of phase $i$  in a semi-synchronous execution of $\NonsynchAlgo$ under $\NonsynchGraph$ and a weakly-fair scheduler. ({Theorem~\ref{thm:outcomes}})
         
         \item{\bf (Finite termination) If the synchronous algorithm always terminates in finite time, so do the simulated executions}: We show that the 
         synchronous execution of $\synchAlgo$ terminates in finite time for $\synchGraph$ if and only if the semi-synchronous execution of $\NonsynchAlgo$ terminates in finite time for $\NonsynchGraph$, which implies the condition.
         (Theorem~\ref{thm:liveness})
    \end{itemize}
         While the {\em non-triviality} requirement rules out trivial solutions, e.g., in which $\NonsynchGraph$ always contains no edges regardless of the actual dynamics, our transformation actually satisfies a {\bf strong non-triviality requirement}: 
     Any edge $(u,v)$ that persists long enough for both nodes $u$ and $v$ to be activated at least once (in two separate stages), during a phase $i\geq 0$ of the semi-synchronous execution {\em must}  be part of $\synchGraph$ during the $i$-th step of the synchronous execution. 
     In particular, if 
     $\NonsynchGraph$ is static, then the simulation guarantees that $\synchGraph = \NonsynchGraph$. 
     
    We assume the {\sc Pull} model of communication (see, e.g.,~~\cite{besta17-pushpull}), with the addition of a {\em 1-bit multi-writer atomic register at each edge port} of every node. Each node is assumed to have a {\em disconnection detector}, as in~\cite{daymude2022-mutex, iyer2003dual, russell2001ethtool}, which indicates which ports experienced edge disconnection since the last time the node was activated. 
    An
    extension 
    of the {\sc Pull} model is necessary since we 
    the standard {\sc Pull} model (and similarly the standard {\sc Push} model) is not powerful enough to support a deterministic synchronizer in our setting, 
     even if nodes have unique ids (and disconnection detectors), as we show in Section~\ref{sec:analysis}.
The \synchronizer -Synchronizer operates with a multiplicative memory overhead at the nodes that is 
logarithmic on the runtime of the underlying synchronous algorithm being simulated: In particular, if the synchronous algorithm terminates in polynomial time (on the number of nodes), the memory overhead is poloylogarithmic.

 Lastly, we present two applications of the synchronizer: One is a classic application for porting  a synchronous dynamic network algorithm for maintaining spanning forests under arbitrary edge dynamics~\cite{barjon_casteigts_chaumette_johnen_neggaz_2014} to a semi-synchronous environment.
 The other, for the minority opinion dynamics presented in~\cite{clementi2024}, allows for an exponential improvement in semi-synchronous runtime when compared to semi-synchronous executions without the synchronizer.

  In summary, we make the following contributions: {\bf (1)} we extend the definition of synchronizers to networks with  arbitrary edge dynamics; {\bf (2)} we present the first synchronizer from the semi-synchronous to the synchronous model in a network with arbitrary edge dynamics; and {\bf (3)} we present non-trivial applications of the proposed synchronizer to existing algorithms.

  The remainder of this paper is organized as follows. Section~\ref{sec:related} presents related work. Section~\ref{sec:model} presents the system model, and Section~\ref{sec:alg} presents the \synchronizer -synchronizer. Section~\ref{sec:analysis} presents the correctness proofs and the analysis of the algorithm. Section~\ref{sec:appl} discusses the applications of the work, and Section~\ref{sec:conclude} concludes the paper.

\section{Related Work}\label{sec:related}

 \begin{table*}[tb]
    \caption{Comparison to other synchronizers}
    \begin{adjustbox}{width=\columnwidth,center}
    \begin{tabular}{c c c c c c}
    \hline
    \mbox{\textbf{Protocol}} & \mbox{\textbf{Year} } & \mbox{ \textbf{Mapping} } & \mbox{ \textbf{Network Dynamics} }  & \mbox{ \textbf{Anonymous} } \\
        \hline
    $\alpha, \beta,\gamma$-Synchronizers~\cite{awerbuch1985-complexity} & 1985 & Sync to Async & static  & No \\
    Afek et al.~\cite{afek1987-applying} & 1987 & Sync to Async & dynamic with eventual-quiescence & No \\
    Awerbuch and Sipser~\cite{awerbuch1988-dynamic}& 1988 & Sync to Async & dynamic  with $t_\pi$-stabilization &  No \\
    Awerbuch and Peleg~\cite{awerbuch1990-network} & 1990 & Sync to Async & static &  No \\
    Awerbuch et al.~\cite{awerbuch1992-adaptingasyncnetworks} & 1992 & Sync to Async & dynamic with eventual-quiescence & No \\
    $\zeta$-Synchronizer~\cite{shabtay1994-synchronizerlowmemoverhead} & 1994 &  Sync to Async & static   & No \\
    $\eta_1,\eta_2,\theta$-Synchronizers~\cite{shabtay1994-lowcomplex} & 1994 & Sync to Async & static   & No\\ 
    $\sigma$-Synchronizer~\cite{pandurangan2020-message} & 2020 &  Sync to Async  &  complete, static    & No \\ 
    Ghaffari and Trygub~\cite{ghaffari2023-nearoptimal} & 2023 & Sync to Async & static  & No\\ 
    \synchronizer -Synchronizer (this paper) & 2025 & Sync to Semi-sync  &  {\bf Arbitrary Dynamics} & {\bf Yes} \\   \hline
    \end{tabular}    \label{tab:related}
    \end{adjustbox}
    \end{table*}

Table~\ref{tab:related} summarizes some main works on synchronizers and their assumptions about network dynamics.       
    The first work on synchronizers was by Awerbuch~\cite{awerbuch1985-complexity}, who defined the problem and presented three synchronizers in a static network with trade-offs between time and message complexity. In his approach, to simulate a synchronous round, a node waits for confirmation that its neighbors finished simulating a previous round before advancing to the next round. This is the approach we take in this paper, but it is made more difficult by the arbitrary network dynamics. 

    The assumption of a static network was relaxed by Afek et al.~\cite{afek1987-applying} who assume a dynamic network with eventual quiescence. This was further relaxed by Awerbuch and Sipser~\cite{awerbuch1988-dynamic} who assume a dynamic network, but only require local quiescence. The idea is that nodes whose local neighborhoods stabilize should not have to wait for the whole network to stabilize before producing an output. This is especially relevant for protocols with short execution times.  Further improvements in execution time for both static and dynamic networks was achieved by follow-up work. 
    In contrast to previous deterministic synchronizers, Awerbuch et al.~\citep{awerbuch1992-adaptingasyncnetworks} propose a randomized synchronizer with polylogarithmic time- and message-complexity overheads for asynchronous dynamic networks. This was a significant improvement over previous synchronizers that incur linear blow-up in either running time, message count, or space requirements. The synchronizer we propose in this paper is deterministic. 
    
    Other synchronizers designed for static networks also assume the presence of node identifiers, such as the $\zeta$-Synchronizer—a modification of $\gamma$—proposed by Shabtay et al.~\cite{shabtay1994-synchronizerlowmemoverhead} that eliminated the need for temporary storage buffers. The authors also introduce the $z$-partition algorithm to reduce the number of external edges connected to each node, thus reducing the memory overhead of $\zeta$. This algorithm optimally balances memory overhead and time complexity. 
    
    More recently, Pandurangan et al.~\cite{pandurangan2020-message}, studying lower bounds for monte-carlo algorithms in a synchronous setting, proposed the $\sigma$-Synchronizer, which achieves only logarithmic simulation overheads in both time and message complexity relative to the number of synchronous rounds. The synchronizer is deterministic, but requires unbounded local memory and assumes a complete (static) network. Ghaffari and
    Trygub~\cite{ghaffari2023-nearoptimal} introduced the first deterministic distributed synchronizer with near-optimal time and message complexity overheads, but assume a static networks with unique node identifiers. 
    
    Our work is significantly different from all these previous works, as we allow the network topology to change arbitrarily without assuming eventual stabilization or unique node identifiers.  In this paper, we consider a new system model with stronger adversarial assumptions compared to what was considered in previous works. Our emphasis is on showing the possibility of having a synchronizer in this challenging setting, and we do not attempt to obtain optimal complexity for the synchronizer. That is the subject of future work.
        
\section{Model} \label{sec:model}
    We consider an edge-dynamic network that we formalize using a time-varying graph \cite{casteigts2018-journeydynamicnetworks}. 
    A time-varying graph is represented as $\mathcal{G} = (V, E, T, \rho)$ where $V$ is the set of nodes, $E$ is a (static) set of undirected pairwise edges between nodes, with a  lifetime $T\subseteq \mathbb{N}$.
    A \textit{presence} function $\rho : E \times T
    \to \{0,1\}$ indicates whether or not a given edge exists at a given time. We define a \textit{snapshot} of $\mathcal{G}$ at time $t$
    as the undirected graph $G_t (V,E_t)$, where $E_t=\{e \in E : \rho(e, t) = 1\}$. At any time, we denote the \textit{neighborhood} of a node $u \in V$  by $N_t(u)$. 
    For $t \geq 0$, the $t$-th \textit{stage} lasts from time $t$ to the instant just before time $t+1$; thus, the communication graph in stage $t$ is $G_t$. While implementing our synchronizer, we refer to these stages in the synchronous environment as \textit{phases}.
    
   Each node has a set of ports it uses to connect to its neighbors, numbered from $\{1, \ldots, \Delta\}$, where $\Delta$ denotes the maximum degree, over time, of the dynamic graph. An edge $(u, v)$ that is present at some stage $t$ will be connected to 
   specific ports assigned at both $u$ and $v$ that do not change while the edge remains present in the dynamic graph.
   To preserve anonymity among nodes and their neighbors, each node $u$ associates with its neighbors through port labels $\ell \in \{1, \ldots, \Delta\}$. Note that different nodes may connect to $u$ through the same port at different times, and a node may connect to different ports of $u$ over time, with the restriction that at most one node is connected to any given port of $u$ at any point in time. We assume that each node is equipped with a mechanism to detect disconnections on its ports~\cite{daymude2022-mutex, iyer2003dual, russell2001ethtool}. When an edge connected to a node $u$ is disconnected, the {\em disconnection detector} adds the corresponding port label to a temporary set $D$. The set $D$ is reset to $\emptyset$ at the end of each activation of
    $u$. 
 
 The memory of a node can be split into three groups---\emph{main state}, \emph{port state}, and \emph{neighbor state}. Each node has a \emph{main state} that it stores in its main memory, which consists of the variables a node maintains that are not directly associated with any port.
For each port there is a \emph{port state}, which stores variables directly associated with the port (for example, the local port label), and a \emph{neighbor state} for the variables that it pulled from the neighbor connected to the port. For convenience in the narrative, we may refer to a node $v$ connected to $u$ through port $\ell$ and use $u.\ell(v)$ to denote the {\em port label at node $u$ corresponding to the edge $(u,v)$}.
    We also use $u.x$ to denote the value of {\em variable $x$ at $u$}.
    
    An algorithm, including the synchronizer, is composed of multiple actions of the form $\langle label\rangle: \langle guard\rangle \to \langle operations\rangle$. An {\em action is  enabled} if its guard (a boolean predicate) evaluates to true, and a {\em node} is said to be {\em enabled} if it has at least one enabled action. The scheduler controls the stages when an enabled node is activated and picks exactly one enabled action at the node to execute at the given stage.
    Our synchronizer algorithm will ensure that we only have one enabled action per node at any stage $t$. 
    
    We assume the {\sc Pull} model of communication (see, e.g.,~\cite{besta17-pushpull}), where 
    every time a node $u$ is activated at a stage $t$, $u$ can pull the state information 
    of a neighboring node $v$. In Section \ref{sec:analysis}, 
    we show that the classic {\sc Pull} (or {\sc Push}) model
    is not powerful enough to support synchronization. 
    Thus, we further equip the {\sc Pull} model with a {\em 1-bit multi-writer atomic register at each edge-port},
    and allow an activated node $u$ at stage $t$ 
    to perform one atomic write onto each multi-writer register of any neighbor node $v$ such that $(u,v)\in E_t$.We show in Section~\ref{sec:analysis} that such an extension of the {\sc Pull} model is necessary for our aimed synchronization in the presence of arbitrary edge dynamics. Each action follows the Read-Compute paradigm, where these two components are executed in this order in lockstep (and writes occur within the compute cycle):
    \begin{itemize}
        \item \textbf{Read:} Pull the main and port state variables from neighbors. 
        \item \textbf{Compute:} Based on the pulled information and the node's current state, compute the next state of the node. 
        That includes updating (writing to) any local variables and allows for writing on the 1-bit multi-writer atomic register at the respective port of each of its current neighbors.
    \end{itemize}
    We focus on \textit{semi-synchronous} concurrency, where in each stage,  any (possibly empty) subset of enabled nodes is activated concurrently, under arbitrary topological changes.
To model this, we assume two adversaries, an {\em adaptive adversary} that controls the edge presence function at each stage (i.e, controls the edge dynamics), and the scheduler adversary---or simply the {\em scheduler}---that controls when nodes are activated. 
We assume a {\em weakly fair} scheduler that activates nodes such that any continuously enabled node is eventually activated (or equivalently, that every enabled node will be activated infinitely often). The only constraint on the edge dynamics adversary is on the maximum degree of the graph, i.e., no node can have more than $\Delta$ neighbors. Hence, $\forall t \in T, \forall u \in V, |N_t(u)| \leq \Delta$. We do not consider faults nor any Byzantine behavior.

      \begin{algorithm}[htb]
        \caption{$\delta$-Synchronizer} \label{alg}
        \begin{algorithmic}[1] 
        \State \textsc{Handshake}: $(\texttt{synch}=0) \lor (\exists \ell \in P\setminus \tilde D
        \suchthat \ell.\texttt{block} = 0)  \to$
        
            \Indent

                    \If{$\texttt{synch}=0$} \BComment{If this is the initialization stage of synchronization} 
                        \For{$\ell \in N(u)$} \BComment{for a phase $i$, \mbox{ } pull information from all neighbors}
                            \State $X(\ell)\gets$ 
                            {\sc Pull}$(\ell)$  
                            \BComment{and initialize $\tilde D$,\mbox{ } the set of disconnected \mbox{ } ports in}\label{algline:pullInfo}
                        \EndFor
                         \State $\tilde D\gets \phi$
                      \BComment{the current phase to the empty set. \hspace{14.5ex}}

                         \State $P\gets \{\ell\in N(u)\suchthat$  \BComment{Valid neighbors of $u$ to be considered in phase $i$'s simulation are: }
                         \Indent
                         \Indent
                         \State $[(X(\ell).\texttt{phase}<\texttt{phase})] \;\lor$ \BComment{(1) neighbors that will catch up with phase $i$, or} \label{algline:beginSetP}
                         \State $[(X(\ell).\texttt{phase}=\texttt{phase}) \;\land$ \BComment{(2) neighbors simulating  phase $i$  that have either}
                         \State $\;\;(X(\ell).\texttt{synch} = 0) \; \lor$  \BComment{(2a) not initialized their valid neighbors, or}
                         \State $\;\;(X(\ell).\texttt{port}\in 
                        X(\ell).P\setminus (X(\ell).\tilde D \cup 
                        X(\ell).D)))]\}$   \BComment{(2b) also consider $u$ as a valid}
                        \Statex{} \BComment{persistent (i.e., not disconnected) neighbor  
                        in phase $i$.}\label{algline:endSetP}
                        
                        \EndIndent
                        \EndIndent
                        
                        \State $\texttt{synch}=1$ 
                        \BComment{Flag end of the initialization stage for phase $i$ of the simulation.} \label{algline:setSynch}
                    \Else \BComment{If $u$ has finished initialization stage of the simulation of phase $i$,}
                        \For{$\ell \in P \setminus (\tilde D \cup D)$} \BComment{For every persistent valid neighbor in the phase:}
                        
                            \If{$X({\ell}).\texttt{phase}<\texttt{phase}$} \BComment {(a) Pull updated state from neighbors that are} \label{algline:checkNeighborLower}
                            \State $X(\ell) \gets$ {\sc Pull}$(\ell)$  \BComment{running behind  and not yet simulating phase $i$.} \label{algline:lowerPhasePull} 
                            \Else \BComment{Otherwise, they must be simulating phase $i$, so}
                                \State $X(\ell).\texttt{ack} \gets$ {\sc Pull}$(\ell).\texttt{ack}$  \BComment{(b)  pull only neighbor's \hspace{0.5mm} ack value for phase $i$.}  
                            \EndIf
                        \EndFor
                        \State $\tilde D \gets \tilde D \cup D$ \BComment{Add disconnections since last activation to the set $\tilde D$.}
                    \EndIf
             
                \Statex 
                \BComment{After pulling state and ack information, attempt edge agreement.}
                \For{$\ell \in (P \setminus \tilde D)$ 
                
                $\suchthat[(X({\ell}).\texttt{phase} = \texttt{phase}) \land 
                (\ell.\texttt{block} = 0) ]$} 
                \label{algline:phaseEqualLoop}
                        \If{$X(\ell).\texttt{ack} = 1$}   \BComment{If a neighbor has already initiated edge agreement, block} \label{algline:neighborAck}
                            \State{{\sc Block}$(\ell)$} \label{algline:setNeighborBlock} \BComment{the edge by simultaneously setting the \texttt{block} flag on neighbor's port} 
                            \State{$\ell.\texttt{block}=1$} \BComment{ as well as on own port  connected to edge to 1.} \label{algline:setOwnBlock}
                        \Else \BComment{Otherwise, indicate that state information for neighbor's} 
                            \State{$\ell.\texttt{ack} = 1$} 
                            \BComment{simulation of phase $i$ has been pulled.} 
                            \label{algline:setAck}
                        \EndIf  
                        \EndFor
            \EndIndent
             \Statex{} 
        \State \textsc{ExecuteSynch:} $(\texttt{synch} = 1) \land (\ell.\texttt{block} = 1$, $\forall \ell \in P\setminus \tilde D) \to$ \label{algline:executeSynch} 
            \Indent
                \State $F \gets \{\ell\in P \suchthat
                (\ell.\texttt{block} = 1) \}$ \BComment{$F$ represents the final  set of agreed neighbors of $u$ in phase $i$} 
                \State Run enabled action of (synchronous) algorithm $\synchAlgo$ with respect to  $\{X(\ell) \suchthat \ell \in F\}$
                \label{algline:actionA}
                \State $\texttt{phase} = \texttt{phase} + 1$ 
                \label{algline:phaseincr}
                \State $\texttt{synch} = 0$ \BComment{Clean up variables}
                \For{$\ell \in \{0,\ldots,\Delta\}$} \BComment{for new phase.}
                    \State $\ell.\texttt{ack} = \ell.\texttt{block} = 0$
                \EndFor
            \EndIndent 
        \end{algorithmic}
    \end{algorithm} 
    
\section{The \synchronizer-Synchronizer}
\label{sec:alg}
In this section, we describe the \synchronizer -Synchronizer. While the objective remains to adapt synchronous network algorithms to semi-synchronous dynamic environments, unlike existing synchronizers, we do not assume that the dynamic network eventually stabilizes.

The \synchronizer -Synchronizer (Algorithm~\ref{alg})
    consists of two actions, \textsc{Handshake} and \textsc{ExecuteSynch}, with exactly one of them being enabled for a node at any stage. 
    Recall that the synchronizer works by showing the equivalence between a semi-synchronous execution of \synchronizer ($\synchAlgo$) under an arbitrary dynamic network $\NonsynchGraph = \{G_0, G_1,\ldots 
    \}$ and a synchronous execution of $\synchAlgo$ under a (potentially different) arbitrary dynamic network $\synchGraph=\{H_0, H_1,\ldots 
    \}$.

    \begin{table*}[htb]
        \centering
        \caption{The notation, domain, initialization, and description of the local variables used in the algorithm by a node $u$. All variables are {\em multi-reader, single-writer atomic registers}, with the exception of the \texttt{block} flag, which is {\em multi-writer}.}
        \label{tab:variables}
        \begin{adjustbox}{width=\columnwidth,center}
        \begin{tabular}{cp{70pt}cp{208pt}}
            \toprule
            \textbf{Var.} & \textbf{Domain} & \textbf{Init.} & \textbf{Description} \\
            \midrule
            \texttt{synch} & \{0,1\}& 0 & A flag that indicates if the current neighborhood info has been pulled for current phase.
            \\

            \texttt{phase} & $\mathbb{N}$ & 0 & A node's own phase number. \\
    
            $\ell.\texttt{block}$ & \{0, 1\} & 0 & The \texttt{block} flag value specific to port $\ell$. This is the only multi-reader, {\em multi-writer} atomic register. \\
    
            $\ell.\texttt{ack}$ & \{0, 1\} & 0 & The acknowledgment flag value specific to port $\ell$. \\
    
            $\ell.\texttt{port}$ & $\{0,\ldots , \Delta-1\}$ & $\ell$ & This variable is in port $\ell$'s memory and signals the port label. \\

            $P$ & $\subseteq N(u)$ & $\emptyset$ & The set of valid ports
            for current phase. \\ 

            $D$ & $\subseteq \{0,\ldots , \Delta-1\}$ & $\emptyset$ & The set of ports that have become disconnected since last time node $u$ was activated. \\

            $\tilde D$ & $\subseteq \{0,\ldots , \Delta-1\}$ & $\emptyset$ & The set of ports that have become disconnected 
            throughout current phase. \\

            $F$ & $\subseteq P$ & $\emptyset$ & The set of valid ports that will be considered for the emulation of the synchronous algorithm by node $u$ during current phase. \\

            $X(\ell)$, $X(\ell).x$ & varied & varied & $X(\ell)$ denotes 
            all of $v$'s main and $v.\ell(u)$-port state variables pulled through port $\ell$ from a neighbor $u$. If indexed by $x$, it denotes the value of the pulled variable $x$; we omit 
            the reference to the port 
            $v.\ell(u)$ for the port variables. \\
            \bottomrule
        \end{tabular}
        \end{adjustbox}
    \end{table*}
    
In a nutshell, the \synchronizer -Synchronizer works by assigning a phase number to each node $u$, stored in the phase variable at $u$, denoted as $u.\texttt{phase}$, which is incremented every time $u$ performs an action of the synchronous algorithm $\synchAlgo$ (Lines~\ref{algline:actionA}-\ref{algline:phaseincr}) it simulates. 
    More concretely, at the end of phase $i$ for node $u$, $u$ will emulate the $i$-th synchronous step of $\synchAlgo$ during an execution of 
    the action {\sc ExecuteSynch} and increment its phase number to $i+1$ right afterwards. Before node $u$ can get to executing an action of $\synchAlgo$ during phase $i$, it must first agree on which (undirected) edges to consider on the $i$-th snapshot $\synchGraphSnap_i$ 
    of the dynamic graph $\synchGraph$ that sets the simulated synchronous execution of $\synchAlgo$. Since the graph $H_i$ will be defined locally by the nodes $u\in V$, one needs to ensure that a node $u$ will include another node $v$ in its neighborhood of $H_i$ if and only if $v$ also includes $u$ in its $H_i$-neighborhood---i.e., $H_i$ is a consistent undirected graph\footnote{A distinction could be made between undirected and bi-directional graphs, but this is not directly relevant in our context.}.
    This is done 
    through the {\sc Handshake} action, as we will explain in more detail below.

There are several variables that the \synchronizer-Synchronizer maintains at a node $u$, some that directly pertain to the state variables of algorithm $\synchAlgo$---which we call the {\em $\synchAlgo$-state} of $u$---and some that are used for the synchronization per se, such as \texttt{phase}, \texttt{synch}, $\tilde D$, etc. We describe the purpose of each variable below and in Table~\ref{tab:variables}. Note that the $\synchAlgo$-state variables of a node $u$ 
    do not change while $u.\texttt{phase}=i$,
    until $u$ gets to its execution of the $i$-th action of $\synchAlgo$ in Line~\ref{algline:actionA}, which will conclude its phase $i$.
    
During the handshake for the simulation of phase $i$, a node $u$ maintains a set
of potential neighbors (identified by the ports to which they are connected) to be included as neighbors in the simulated execution, which is initially equal to the set $u.P$ of all valid neighbors at the start of the handshake,
when $u.\tt{synch} = 0$. The set of potential neighbors
can "lose" elements in  other stages of the handshake for phase $i$ due to edge disconnections, 
but does not gain new elements. The crucial property to maintain is {\em edge consistency}: At the end of the simulation of phase $i$, the {\em final set} of neighbors of $u$, $u.F$, which is a subset of $u.P$, 
contains node $v$ if and only if the set $v.F$ contains node $u$. The sets $u.F$, for all nodes $u\in V$, determine the set of (undirected) edges in the simulated graph $H_i$ for phase $i$. 

Initially, the set $u.P$ contains the following nodes: 
\begin{description}
    \item[1.] (Running behind) Neighbors that have not yet started simulating phase $i$. These nodes can be included because node $u$ will wait for them to catch up during the next stages in $u$'s handshake for phase $i$ (unless they disconnect from node $u$),
    
    \item[2a.] (Concurrent nodes in initialization stages) Neighbors that have started simulating phase $i$ but have not yet finished the initialization stage of the handshake.
    These nodes will also be able to see that node $u$ has not yet finished its initialization stage of the simulation.
    
    \item [2b.] (Concurrent nodes past initialization  stages) Neighbors that completed the initialization of the simulation of phase $i$ and are still considering node $u$ as a neighbor in their handshake of phase $i$. These nodes must have started the handshake of phase $i$ before node $u$ but have not detected a disconnection in the edges linking them to node $u$.
\end{description} 
To calculate the set $u.P$, a node pulls information from its neighbors and determines which nodes fall into one of the categories above (Lines \ref{algline:beginSetP}--\ref{algline:endSetP}). 
At the end of the first stage of the handshake, node $u$ sets its synch flag to 1 to signal the end of its initialization for phase $i$.

To ensure {\em edge consistency}, each node $u$ keeps track of all ports that have experienced a disconnection since the start of the simulation of phase $i$ in its set $u.\tilde D$ (initially empty). All neighbors connected through ports in $u.\tilde D$ are removed from consideration
to be included in the simulated graph $H_i$.
Indeed, if there is a disconnection of an edge $(u,v)$ in a given stage, then nodes $u$ and $v$ will each add the respective port connected to edge $(u,v)$
to its own set $\tilde D$ in the first stage in which the node is active after the disconnection, and the edge $(u,v)$ and will not be included in the simulated graph $H_i$ by either node. Since nodes are anonymous, 
any edge that later connects  to the port of either $u$ or $v$ earlier connected to $(u,v)$---including, potentially, a reconnection of $u$ and $v$---will be ignored in phase $i$.

A node $u$ checks if the neighbors running behind in the simulation have caught up to phase $i$ (Line~\ref{algline:checkNeighborLower}) or, if the neighbor $v$ has caught up, $u$ only needs to update the respective ack value (since it already pulled the state information from $v$ for phase $i$ the last time it pulled from $v$ in Line~\ref{algline:lowerPhasePull}). 
Note that 
when the check is done in Line~\ref{algline:checkNeighborLower}, if the neighbor is not behind, it must be simulating the same phase: The reason is that if the node was previously behind it could not have advanced beyond phase $i$ without first executing the ack/block exchange and the first (and only) time that is done is when the two nodes are in phase $i$.

A node $u$ completes the 
handshake of phase $i$ when it determines that all edges linking it to a tentative neighbor in $u.P\setminus u.\tilde D$ are blocked (Line~\ref{algline:executeSynch}), implying that such an edge $(u,v)$ will be considered to be in $H_i$ by both $u$ and $v$. The blocking of such an edge $(u,v)$ can be initiated by 
$u$ itself or by node $v$, but in either case it will result in the block flags on the respective ports at $u$ and $v$ being set to 1 during the current stage (while the edge is up).  Whichever node blocks the edge-ports, say $u$, must have detected that the other node ($v$) set its ack to 1 (Lines~\ref{algline:neighborAck}-\ref{algline:setOwnBlock}), acknowledging that $u$ 
and $v$ are both simulating phase $i$, and that $v$ has pulled  $u$'s phase $i$ state information prior to the blocking ($u$ has also pulled $v$'s phase $i$ state during the current stage). Note that there are blocked  edges 
in phase $i$ that disconnect later in the phase: Those edges will be in the respective sets $F$ and thus also in $H_i$.
 
\section{Analysis}
\label{sec:analysis}
 We now prove the formal guarantees of the \synchronizer -Synchronizer.   
    For clarity of exposition, we annotate 
    phase-related variables (e.g., $P, F$) with the {\em superscript} $i$ to indicate their state {\em once they are initialized in phase $i$},
    and other state-related variables (e.g., $\texttt{phase},\texttt{ack},\texttt{block},D,\tilde D$) with the {\em subscript} $t$ to indicate their state {\em at the start of stage} $t$.
    We define the notion of a {\em valid persistent neighbor} of node $u$ for a phase $i$ at stage $t$, which informally can be described as a valid neighbor $v$ of $u$ at the start of phase $i$ which stays connected to $u$ through to stage $t$.

    \begin{definition}[Valid Persistent Neighbors]
      For two nodes $u$ and $v$ and $i\geq 0$,
      \begin{enumerate}
          \item[\bf a.] If $v.phase_{t{(u,i)}} \leq u.phase_{t{(u,i)}}=i$ and $v$ was continuously connected to one of the ports $\ell$
      of $u$ for all $r\in [t_{(u,i)},t]$ (i.e., $v \in (u.P^i\setminus u.\tilde D_t)$). Then $v$ is said to be a  {\em valid persistent neighbor} of $u$ at stage $t$ of phase $i$. 
    
      \item[\bf b.] Moreover,
      if there exists a $t\in [\max \{t_{(u,i)},t_{(v,i)}\},\min \{t_{(u,i+1)},t_{(v,i+1)}\}]$ such that $u$ is a persistent neighbor of $v$ and $v$ is a persistent neighbor of $u$ at stage $t$ in phase $i$, then edge $(u,v)\in E_r$  
        for all $r\in [\max\{t_{(u,i)},t_{(v,i)}\},t]$, 
        and $u$ and $v$ are {\em mutual valid persistent neighbors} 
        for phase $i$ (at stage $t$). 
      \end{enumerate} 
        \label{definition}
        \end{definition}

        Note that \synchronizer ($\synchAlgo$) always works on the set of valid persistent neighbors at stage $t$, be it when considering the neighborhood of a node in Line~\ref{algline:phaseEqualLoop} or in the guards of the two actions. 
    We use the notation $t_{(u,i)}$ to denote the first stage $t$ at which a node $u$ is activated with $u.\texttt{phase}_t=i$.

We start by showing, in the next lemma, that for any edge $(u,v)$ and any $i\geq 0$, $u$ and $v$ will both either have their respective port \texttt{block} variables set to 1 for phase $i$, if edge $(u,v)$ persists for long enough in that phase, or they will both have their \texttt{block} variables set to 0. In other words, they either agree to consider undirected edge $(u,v)$ for phase $i$, or not. This guarantees that each snapshot graph $\synchGraphSnap_i=\{(u,v)\in E \suchthat u.\ell(v).\texttt{block}^i=v.\ell(u).\texttt{block}^i=1\}$ of $\synchGraph$ at stage $i$ is a valid undirected graph from the local point-of-view of each $u.F^i$.

     \begin{lemma} \label{lem:consistentNeighborhood}
        For any two nodes $u$ and $v$,  $u$ perceives  the undirected edge $(u,v)$ as present in phase $i$ (i.e., $(u,v) \in \synchGraphSnap_i$) 
        if and only if $v$ also perceives $(v,u)\in \synchGraphSnap_i$. In other words, $u\in v.F^i$ if and only if $v\in u.F^i$.
        Moreover, if $t$ is the stage when node $v$ executes {\sc ExecuteSynch} for phase $i$, then $v.P^i\setminus v.\tilde D_t \subseteq v.F^i\subseteq v.P^i$, i.e.,  $v.F^i$ is a superset of all \emph{valid persistent neighbors} of node $v$ at the end of phase $i$.
    \end{lemma}
    \begin{proof}
         In order to prove the lemma, we first show that the local neighborhood view of each node throughout phase $i$ leads to a consistent view of the {\em undirected} graph $\synchGraphSnap_i$, i.e., we need to show that $v\in u.F^i$ if and only if $u\in v.F^i$. Let $u$ and $v$ be two {\em mutual  valid persistent neighbors} for phase $i\geq 0$ and let $t$ be the latest stage that satisfies Definition~\ref{definition}.b. Recall that  
        $u.\ell(v)$ and $v.\ell(u)$ denote the ports of $u$ and $v$, respectively, connected to the edge $(u,v)$ during $[\max \{t_{(u,i)},t_{(v,i)}\},
        t]$. Note that at stage $t$, 
        the phases of both $u$ and $v$ must be equal to $i$ and $u$ and $v$ are the neighbors that have been connected through $v.\ell(u)$ and $u.\ell(v)$ from the beginning of phase $i$ at $v$ and $u$, respectively, i.e., $u\not\in v.{\tilde D}_t$ and $v\not\in u.{\tilde D}_t$.  
        A node $v$ will be in the set $u.F^i$  if at a stage $t$, 
        $u.\ell(v).\texttt{block}$ is set to 1.
        Without loss of generality, assume that node $u$ is activated by the scheduler at stage $t$ and that it actively sets $u.\ell(v).\texttt{block}=1$.
        Following the algorithm, it can only do so right after atomically setting the multi-writer register
        $v.\ell(u).\texttt{block}=1$ at node $v$ (via the atomic {\sc Block} operation in Line~\ref{algline:setNeighborBlock}). (If $v$ is also activated at stage $t$ and attempts concurrently to set the corresponding \texttt{block} flags, one on its own memory and the other at the $u.\ell(v)$ port memory of $u$, we assume that the concurrent writes of $u$ and $v$ will succeed in writing the common value of 1 in the respective \texttt{block} variables.)

        The \texttt{block} flags for all ports of a node $u$ are set to 0 at the start of a new phase and
        if a \texttt{block} flag is set to 1 during a  phase $i$, it never goes back to 0 during the same phase. Thus we show that 
        the only way a node $v$ can be placed in the set $u.F^i$  (and, similarly $u$  in $v.F^i$) is if $u$ and $v$ are mutual valid persistent neighbors 
        at some stage during phase $i$.
        If there exists an edge $(w,z)\in E_r$ with $i=w.\texttt{phase}_r<z.\texttt{phase}_r$ then the respective \texttt{block} flags for the ports connecting $w$ and $z$ in $E_r$ cannot be set to 1 in stage $r$, since
         (i) 
         $z$ will not be part of $w.P^i$ and hence will not be considered by $w$ in phase $i$, 
         and (ii) $z$ also
        cannot consider $w$ in the loop at Line~\ref{algline:phaseEqualLoop} of the \synchronizer-Synchronizer algorithm during  stage $r$ since the phase numbers of $w$ and $z$ are not equal. 

        Moreover, the \texttt{block} flags for an edge $(w,z)$ cannot be set to 1 by node $z$ in a stage $r$ with $w.\texttt{phase}_r=z.\texttt{phase}_r=i$ if (i) $w\not\in z.P^i$
        or (ii) $w\in (z.\tilde D_r \cup z.D)$ (i.e., there was a disconnection in the port $z.\ell(w)_r$ on or before stage $r$, during phase $i$). Putting it altogether, an edge $(w,z)$ may only be considered to be added to $\synchGraphSnap_i$ during stage  $r$ 
        if $w$ and $z$ are mutual valid persistent neighbors in phase $i$ at stage $r$ (see Definition~\ref{definition}).

        It remains to argue that if $t$ is the stage when node $v$ executes {\sc ExecuteSynch} for phase $i$, then $v.P^i\setminus v.\tilde D_t \subseteq v.F^i\subseteq v.P^i$. Given the definition of the sets $F$, $v.F^i$ is trivially contained in $v.P^i$. To show that $v.P^i\setminus v.\tilde D_t \subseteq v.F^i$, it suffices to note that {\sc ExecuteSynch} will only be enabled for $v$ at a stage $t$ (during phase $i$ for $v$) if $v.\ell(u).\texttt{block}=1$ for all $u\in v.P^i\setminus v.\tilde D_t$.
        \end{proof}

        Recall that the {\em $\synchAlgo$-state of $u$} at phase $i\geq 0$ is equal to the configuration of all state variables maintained by $\synchAlgo$ (namely, the configuration of all of $u$'s  state variables other than 
         $\{ u.\texttt{phase}, u.\texttt{synch}, u.P, u.F, u.\tilde D, u.D; u.\ell.\texttt{block}, u.\ell.\texttt{ack}, u.\ell.\texttt{port} \suchthat \forall \ell\in [\Delta] \}$ and any pulled variables stored at $u$) at the start of phase $i$. 
         We further let $u.X_{\synchAlgo}(\ell)$ denote the $\synchAlgo$-state pulled by $u$ from the neighbor connected through port $\ell$, and $u.X_{\synchAlgo}(\ell)^i$ denote the
         {\em last value $X_{\synchAlgo}(\ell)$ that node $u$ pulls} through port $\ell\in F^i$ during phase $i$.

 We now show that at the start of the $i$-th action execution of $\synchAlgo$, $i\geq 0$, by node $u$ within the \synchronizer-Synchronizer algorithm, $u$ has a snapshot of the $\synchAlgo$-state in phase $i$
 of all its valid neighbors $v$ in $\synchGraphSnap_i$ that corresponds to the $\synchAlgo$-state of the neighbor $v$ in $\synchGraphSnap_{i}$ right after the
     $i$-th synchronous step in a synchronous execution of algorithm $\synchAlgo$ according to $\synchGraph$.
 
    \begin{lemma}

     For any node $u$ and phase $i\geq 0$, the value of $u.X_{\synchAlgo}(\ell)^i$, 
     for each port $\ell\in u.F^i$ corresponds to the $\synchAlgo$-state of the valid neighbor $v$ of $u$ connected through port $\ell$ at the start
     of the $i$-th execution of {\sc ExecuteSynch} by $v$. Moreover, 
     $u.X_{\synchAlgo}(\ell)^i$ corresponds to the $\synchAlgo$-state of the neighbor $v$ 
     of $u$ in $\synchGraphSnap_{i}$ right before
     the 
     $i$-th 
     step in a synchronous execution of algorithm $\synchAlgo$ according to $\synchGraph$.
     \label{lem:consistentAlgoValues}
     \label{lem:consistentPulledVars}
    \end{lemma}

    \begin{proof} 
         We will prove the lemma by induction on $i\geq 0$. The claim is trivially true for $i=0$, since
         by a node $u$ corresponds to the first action of $\synchAlgo$ that $u$ executes, which cannot rely on any proper initialization of $u$'s neighbors in a synchronized execution, i.e., cannot rely on any particular value of the pulled states of its neighbors in the synchronous execution.
         Hence, the pulled states $u.X_{\synchAlgo}(\ell)^0$ by $u$ for each port $\ell$ in the semi-synchronous execution must be trivially valid.
         
         Assume the claim is true for any phase less than $i\geq 1$. From Lemma \ref{lem:consistentNeighborhood}, we know that the final set $u.F^i$ represents the neighborhood of $u$ in $\synchGraphSnap_i$. As defined earlier, let $t_{(u,i)}$ be the first stage where $u$ starts an action execution already in phase $i$.
        Note that the $\synchAlgo$-state of any node $v$ only changes during an execution of {\sc ExecuteSynch}, so the value of any $\synchAlgo$ variables of $v$ remains unchanged for the duration of the phase until right before an action of $\synchAlgo$ is executed  in  Line~\ref{algline:actionA} of the \synchronizer  -Synchronizer, i.e., any pulled $\synchAlgo$-state of $v$ remains unchanged for the duration of  $[t_{(v,i)},t_{(v,i+1)}-1]$.
 
        By induction, at the start of the $i-1$-th
        execution of {\sc ExecuteSynch} by $u$ (during its phase $i-1$), 
        $u$ had pulled the correct  $\synchAlgo$-states $u.X_{\synchAlgo}(\ell(w))^{i-1}$ of each of its valid neighbors $w$ such that $(u,w)\in \synchGraphSnap_{i-1}$, 
        i.e., $u.X_{\synchAlgo}(\ell(w))^{i-1}$ is equal to the state of the variables of $\synchAlgo$ after the
        $(i-2)$-th
        synchronous step of algorithm $\synchAlgo$ according to $\synchGraphSnap_0,\ldots , \synchGraphSnap_{i-2}$, for each $(v,w)\in H_{i-1}$.
        Thus when a node $u$ executes its $(i-1)$-th action according to $\synchAlgo$ (within {\sc ExecuteSynch}), $\synchAlgo$ will behave exactly as it would in a (deterministic) execution of the $(i-1)$-th synchronous step of $\synchAlgo$ under $H_0,\ldots , H_{i-1}$, resulting in the computation of the correct $\synchAlgo$-states for each node $v$ after executing the $(i-1)$-th synchronous step given $H_0,\ldots , H_{i-1}$, i.e., the state that algorithm $\synchAlgo$ will pull from each node $v$ at the start of the $i$-th synchronous action execution.
        
        As we argued above, the $\synchAlgo$-state of a node $u$ does not change unless $u$ executes an action of $\synchAlgo$ within {\sc ExecuteSynch}. 
        Hence, all that remains to be shown is that at the start of the $i$-th execution of {\sc ExecuteSynch} by $u$ (during its phase $i$), $u$ will have pulled the ${\synchAlgo}$-state for each of its neighbors $w$ such that $(u,w)\in \synchGraphSnap_{i}$ (i.e., all $w\in u.F^i$) at least once while $w.phase=i$.

        Let $v\in u.F^i$, which in turn implies that $u\in v.F^i$.
        If $v.\texttt{phase}_{t_{(u,i)}}>u.\texttt{phase}_{t_{(u,i)}}=i$, $v$ is not a valid neighbor of $u$ in phase $i$, i.e., $v\not\in u.P^i\supseteq u.F^i$, a contradiction. 
        Thus we must have $v.\texttt{phase}_{t_{(u,i)}}\leq u.\texttt{phase}_{t_{(u,i)}}=i$. 
        Let $t$ be the stage in $[t_{(u,i)},t_{(u,i+1)}-1]$ when $v$ executed its $(i-1)$-th {\sc ExecuteSynch} and set $v.phase=i$ at the end of this execution. The first stage when node $v$ starts an action in phase $i$, $t_{(v,i)}$, must then lie within $[t,t_{(u,i+1)}]$. 
        Node $v$ executes a  {\sc Handshake} action 
        at $t_{v,i}$ and sets its \texttt{ack} variable to 1 to indicate that it has executed a {\sc Pull}$(\ell(u))$ (in Lines~\ref{algline:pullInfo} or \ref{algline:lowerPhasePull}) while $u.phase =i$, 
        implying that $v.\ell(u).\texttt{ack}_{t'}=1$, for all $t'\geq t_{(v,i)}$ while $v.\texttt{phase}_{t'}=i$.  From $u$'s point-of-view at stage $t_{(v,i)}$ there are two cases: (i) if $u$ is activated at some stage in $[t+1,t_{(v,i)}]$, then $u$ must have pulled $X_{\synchAlgo}(\ell(v))$ while $v.phase=i$ and hence it will have set its $u.\ell(v).ack=1$ for the remainder of its phase $i$; otherwise (ii) $u$ will be activated next (by weak fairness) at some stage $t''\geq t_{(v,i)}$ and pull $X_{\synchAlgo}(\ell(v))$ while $v.phase=i$ and set $u.\ell(v).ack=1$.
        In case (i), $v$ will be the first to acknowledge that $v.X(\ell(u)).\texttt{ack}=1$ and will set its own $v.\ell(u).\texttt{block}$ and $u.\ell(v).\texttt{block}$ to 1. In case (ii), $u$ will be the first to acknowledge that $u.X(\ell(v)).\texttt{ack}=1$ and will set its own $u.\ell(v).\texttt{block}$ and $v.\ell(u).\texttt{block}$ to 1. This \texttt{ack-block}-handshake ensures that we only set the \texttt{block} flags and thus consider a neighbor node for $F^i$ when both endpoints of the corresponding edge have pulled the correct $\synchAlgo$-state for phase $i$. 
    \end{proof}

Putting it all together, the following theorem states the correctness of the \synchronizer -Synchronizer:

    \begin{theorem}[Correctness]
    For any semi-synchronous execution of \synchronizer($\synchAlgo$) under an arbitrary edge-dynamic graph $\NonsynchGraph$ and a weakly-fair scheduler, there exists an arbitrary edge-dynamic graph $\synchGraph$ such that the $\synchAlgo$-state of each node at the end of each phase $i$ of the semi-synchronous execution of \synchronizer($\synchAlgo$) under $\NonsynchGraph$ and the weakly-fair scheduler, is equal to the state of each node at the end of step $i$ in a synchronous execution of $\synchAlgo$ under $\synchGraph$.
    \label{thm:safety}
    \end{theorem}
    
We will show that, conversely, any synchronous execution of $\synchAlgo$ on an arbitrary edge-dynamics time-graph $\synchGraph$ can be viewed as the execution of $\synchAlgo$ as mandated by the \synchronizer -Synchronizer in a semi-synchronous environment where, in a nutshell, the scheduler 
basically
activates all nodes at every stage and 
the graph $\NonsynchGraph$ mimics the dynamics of $\synchGraph$, by simulating every configuration of $\synchGraph$ over three consecutive stages of $\NonsynchGraph$.

\begin{theorem}[Weak non-triviality]
For any synchronous execution of $\synchAlgo$ under an arbitrary edge-dynamic graph $\synchGraph$, there exists an arbitrary edge-dynamic graph $\NonsynchGraph$ such that the state of each node at the end of each step $i$ of the synchronous execution of $\synchAlgo$ under $\synchGraph$, is equal to the $\synchAlgo$-state of each node at the end of phase $i$ 
    in a semi-synchronous execution of \synchronizer($\synchAlgo$) under $\NonsynchGraph$ and a weakly-fair scheduler, for all $i\geq 0$. 
\label{thm:outcomes}
\end{theorem}

\begin{proof}
  We show that every synchronous execution of $\synchAlgo$ under an arbitrary edge-dynamics time-varying graph $\synchGraph$ can be reproduced by the \synchronizer-Synchronizer in a weakly-fair semi-synchronous execution of \synchronizer($\synchAlgo$)
  for a given 
  time-varying graph $\NonsynchGraph$. For every $i\in [T_{sync}] $, we let $G_{3i}=G_{3i+1}=G_{3i+2}=\synchGraphSnap_{i}$.  For every $0\leq i\leq T_{sync}$, we assume that the weakly-fair scheduler activates all nodes in stage $3i$, activates all nodes with at least one adjacent edge in $E_{3i+1}=E'_i$ in stage $3i+1$, and activates all the nodes in stage $3i+2$ of the semi-synchronous execution of \synchronizer($\synchAlgo$). Starting from stage 0 with all variables properly initialized (in particular, all \texttt{block} flags, phase numbers and synch values set to 0 for all nodes), for every $0\leq i\leq T_{sync}$ and every node $u$, at stage $3i+1$ all the \texttt{block} flags for the ports of $u$ corresponding to an edge $(u,v)\in E_{3i+1}(=E'_{i})$ will be set to 1 (since all nodes will set their synch variables to 1 and all the ports connected to an edge will have their ack flags set to 1 in stage $3i$; the edge set remains unchanged between stages $3i$ and $3i+1$ of the semi-synchronous execution). Thus, every node $u$ will execute {\sc ExecuteSynch} in stage $3i+2$ with $F^i=E_{3i+2}=E'_i$ and the algorithm $\synchAlgo$ progresses in stage $3i+2$ of the semi-synchronous execution of \synchronizer($\synchAlgo$) exactly as it would in the $i$-th synchronous 
  step under $\synchGraph$; at the end of the {\sc ExecuteSynch} execution, each node $u$ will increment its phase number to $i+1$ and reset all the other relevant variables to 0.
\end{proof}

Together, Theorems~\ref{thm:safety} and \ref{thm:outcomes}, show that the \synchronizer -Synchronizer correctly emulates all the possible synchronous executions of algorithm $\synchAlgo$, and only those,  under arbitrary edge--dynamics.
 We now prove the following lemma, which will allow us to show liveness of the algorithm. We measure {\em time} by number of elapsed stages.
    
    \begin{lemma} 
    For any 
    $i\geq 0$, given a stage $r_i$
    such that the smallest phase number of a node at the start of $r_i$ is equal to $i$ (i.e., $\min_{w\in V}\{w.\texttt{phase}_{r_i}\}=i$), it takes a finite amount of time to reach a stage $r_{i+1}$ where the smallest phase number of a node at the start of
    stage $r_{i+1}$ is equal to $i+1$ (i.e., $\min_{w\in V}\{w.\texttt{phase}_{r_{i+1}}\}=i+1$). 
    \label{lem:minphase}
    \end{lemma}
    
    \begin{proof}
    We start by noting that due to weak fairness, the adversary must activate 
        any continuously enabled node in a finite amount of time, and
        given the design of our algorithm, every node always has exactly one enabled action at any given stage. 
            Let $w\in V$ be such a node with $w.\texttt{phase}_{r_i}=i$.
            Each valid persistent neighbor $z$ of $w$ at $r_i$ (i.e., $z\in w.P_{r_i} \setminus w.\tilde D_{r_i}$)
            has a phase number equal to $i=w.\texttt{phase}_{r_i}$ (since it cannot have a lower phase number) and hence, it
            will be considered in the  loop in Lines~\ref{algline:phaseEqualLoop}-\ref{algline:setAck}.
            By Lemma~\ref{lem:consistentNeighborhood} and weak fairness, $w$ and $z$ 
            will be activated
            in finite amounts of time from $r_i$ and 
            will either advance to set both of their \texttt{block} flags corresponding to the ports connected to edge $(w,z)$ to 1 or this edge becomes disconnected before the \texttt{block} flag is set at either of $w$ or $z$ (implying that $z$ enters $w.\tilde D_{r}$ for some $r>r_i$). 
            Once this happens, in finite time, for all valid persistent neighbors of $w$ at stage $r_i$, {\sc ExecuteSynch} will become enabled for $w$ at some (finite) stage $r(w)>r_i$, resulting in $w.\texttt{phase}_{r(w)}$ being incremented to $i+1$. Let $r_{i+1}-1$ be the stage when the last nodes with phase number equal to $i$ at the start of stage $r_i$ increments their phase number. Then at the start of stage $r_{i+1}$ we have that $\min_{w\in V}w.\texttt{phase}_{r_{i+1}}=i+1$. Since there are a finite number of nodes and since each node has a finite number of neighbors at any stage of ${\NonsynchGraph}$, $r_{i+1}-r_i$ must be finite, proving the lemma.
            \end{proof}

    The following theorem, stating the liveness of the \synchronizer -Synchronizer algorithm, is a direct consequence of the fact that the minimum phase number increases by one in a finite amount of time, as proven in Lemma~\ref{lem:minphase}.
    \begin{theorem}[Finite termination] \label{thm:liveness}
        Our synchronizer ensures liveness, since any node will progress in finite time to phase $i$, for any finite $i\geq 0$. Moreover, all synchronous execution of $\synchAlgo$ terminate in finite time if and only if all the semi-synchronous executions of \synchronizer($\synchAlgo$) also do. 
    \end{theorem}
    \begin{proof}
    By Lemma~\ref{lem:minphase}, all nodes have progressed to at least phase $i$ by some stage  $r_i$, as defined in the lemma.
By setting $r_0=0$ and applying Lemma~\ref{lem:minphase} iteratively for $0\leq j<i$, we get an $r_i$ that must be finite, since $i$ is finite ($r_i$ would correspond to a finite sum of finite numbers).  Since the maximum number of phases in any execution of \synchronizer($\synchAlgo$) is equal to the worst-case number of steps that it takes for $\synchAlgo$ to terminate, then the second part of the theorem follows.
    \end{proof}

    Theorems~\ref{thm:safety}, \ref{thm:outcomes} and~\ref{thm:liveness} ensure that algorithm $\synchAlgo$ terminates
    in finite time for any synchronous execution 
    under an (arbitrary) time-varying graph $\synchGraph$ (i.e., $T_{sync}$ is finite) if and only if \synchronizer($\synchAlgo$) correctly emulates $\synchAlgo$ for any (arbitrary) time-varying graph $\NonsynchGraph$ also
    in finite time. The following corollary summarizes these results:

    \begin{corollary} 
    \label{cor:main}
        An algorithm $\synchAlgo$ 
        always terminates in finite  time for any given synchronous execution under arbitrary edge--dynamics if and only if \synchronizer ($\synchAlgo$) correctly emulates $\synchAlgo$ and 
        always terminates in finite time for any given weakly-fair semi-synchronous execution and arbitrary edge dynamics. 
    \end{corollary}

We now note that adding a 1-bit multi-writer atomic register at each port of a node as an extension to the {\sc Pull} model is a {\em necessary minimal extension} of the model that allows for the implementation of a synchronizer that can preserve the set of possible outcomes and finiteness of executions of an underlying synchronous algorithm, since such a task would be impossible under the classic {\sc Pull} or {\sc Push} and arbitrary edge dynamics, as we show below.
In the {\em edge--agreement} problem, suppose two nodes, $u$ and $v$, connected by an edge $(u,v)$ at some stage $t$ of a dynamic graph, want to reach an agreement regarding the final value of some 0-1 variable, which is currently set to 0 at both $u$ and $v$: The rule is that if one of the nodes changes its value to 1 at stage $t$, the other node also has to change its value to 1 at some finite stage greater than or equal to $t$; otherwise, $u$ and $v$ remain with the value of 0 for the variable.
In other words, if a node $u$ decides to consider the edge
$(u,v)$ at some stage, then node $v$ has to also commit to consider the edge
within finite time; otherwise both nodes decide not to consider the edge. 

This is a fundamental building block of any deterministic synchronizer for undirected dynamic graphs that preserves finite time executions, which is our case. Without such a building block, all that the nodes would be able to guarantee to agree on in finite time is to never consider any edge in the network, but that would violate the non-triviality requirement.
In the theorem below, we assume that we have an underlying arbitrary edge--dynamics network, where nodes can have unique ids (i.e., nodes are more powerful than the anonymous nodes considered in the previous sections\footnote{Nodes can be assumed to have access to a disconnection detector instead of node ids, matching the model in this paper, but we chose to present the impossibility proofs assuming node ids since this a more powerful variant of the model.}). 

    \begin{theorem} \label{thm:imposs} 
    In a semi-synchronous 
    dynamic network environment with arbitrary edge dynamics and unique node ids,  
    one cannot solve the edge-agreement problem deterministically (in finite time) under the classic {\sc Pull} or {\sc Push} models.
    \end{theorem}
    
    \begin{proof} 
    Consider a dynamic graph scenario under the {\sc Pull} model, with two nodes $u$ and $v$ connected by edge $(u,v)$ at some stage $t$. Assume w.l.o.g. that $u$ is activated at stage $t$ and
    that node $u$ decides it is safe to propose to $v$ that they should agree on the edge $(u,v)$. Node $u$ therefore sets the value of its edge-agreement variable to 1. Since we follow the Read-Compute mode of execution for our actions, node $v$ will only be able to read the value of 1 proposed by $u$ in a stage higher than $t$.
        Assume that the edge $(u,v)$ disconnected at a stage $t'>t$ and remained disconnected until node $u$ is activated
        again at stage $t''\geq t'>t$. 
        Since $u$  cannot pull any information from $v$ at stage $t''$, $u$ cannot decide whether $v$ has woken up during $[t
        +1, t'-1]$, and has therefore pulled the information on the agreement variable, or not. Thus, $u$ cannot decide whether it should keep its proposed value of 1 (to agree with $v$ in case it pulled this information from $u$ before the edge disconnected) or continue to wait for $v$ to accept its proposition.
        Since we are assuming arbitrary edge dynamics, it may be the case that the edge $(u,v)$ will never reconnect and hence node $u$ will not be able to decide at any finite stage on whether it can consider edge $(u,v)$ or not.

        A similar argument holds for the {\sc Push} model.
    \end{proof}

In an undirected dynamic graph model, synchronization relies on reaching agreement on the presence or absence of undirected edges between each pair of nodes, which leads to the following corollary.

\begin{corollary}
\label{cor:imposs}
Any deterministic synchronizer that relies on building consistent undirected graph snapshots 
cannot be executed in finite time in a non-synchronous environment under arbitrary edge dynamics in
the classic \textsc{Pull} or \textsc{Push} communication models.
\end{corollary}

For completeness, we conclude by stating the memory overhead of our \synchronizer-Synchronizer, which trivially follows from the previous theorems and the fact that the node memory requirements of $\synchAlgo$ is $\Omega(\Delta)$, given that each node maintains $\Delta$ ports:

\begin{corollary}
     Let $R$ denote the worst-case running time for a synchronous execution of $\synchAlgo$. Then the synchronizer requires a memory overhead per node of $\Theta (\Delta + \log R)$. In particular, if $R$ is polynomial in $n$, the memory overhead of the synchronizer will asymptotically add a polylogarithmic in $n$ term on the memory requirements of $\synchAlgo$.
     \end{corollary}

\section{Applications and Extensions}
\label{sec:appl}
Designing algorithms for highly dynamic networks without any assumptions on edge dynamics or eventual stabilization is challenging, even in synchronous settings. 
As a result, application scenarios for our synchronizer have proven to be fewer and more complex.
On the other hand, scenarios with high and unpredictable network dynamics are getting increasingly more common in practice, and practitioners are looking at the benefits of time synchronization in order to better manage the dynamics (see, e.g.,~\cite{google-spanner-2013}).  This work aims to bridge this divide.
We expect this work to spark more interest in the area of highly dynamic networks, without quiescence, and to provide a bridge that would encourage more researchers to work on applications under synchronous dynamic networks that can be ported to non-synchronous scenarios.

In this section, we describe a couple of concrete application scenarios where our synchronizer can make a meaningful impact, as we describe below. 
In a classic application of our \synchronizer -synchronizer, 
we consider the work in~\cite{barjon_casteigts_chaumette_johnen_neggaz_2014}, where they present a synchronous algorithm for maintaining a spanning forest that approximates the minimum possible number of spanning trees, in highly dynamic networks with no quiescence assumptions. Their work assumes the {\sc Push} model.
Given the simplicity and structure of the spanning forest maintenance algorithm, it can be seamlessly converted from the {\sc Push} to {\sc Pull} model.  
By executing this 
algorithm within the \synchronizer-Synchronizer, we extend its applicability to semi-synchronous environments, preserving correctness while enabling analysis of time and communication bounds under weaker assumptions.

Another application is in the context of \textit{minority dynamics}~\cite{clementi2024}, a stateless protocol in which each node samples a random subset of neighbors and adopts the minority opinion observed. The impact of our synchronizer is not just about enabling synchronous algorithms in semi-synchronous environments, but about providing an exponential speed-up in runtime when doing so. 
This dynamic exhibits dramatically different behavior under non-synchronous models when compared to synchronous ones. In asynchronous or sequential activation models---the sequential model being a special case of the semi-synchronous or asynchronous models---the expected convergence time is exponential in the number of nodes. In contrast, under synchronous models such as synchronous gossip, it converges in $O(\log^2 n)$ rounds with high probability. The protocol assumes 
a directed
dynamic communication model, 
eliminating the need for two nodes $u$ and $v$ to agree on the inclusion of edge $(u,v)$ in a given phase, and thus also the need for the 1-bit multi-writer atomic port registers.
Nevertheless, using a directed version of the  \synchronizer-Synchronizer  in this context would bring the run time of the minority dynamics algorithm of~\cite{clementi2024} in semi-synchronous environments under a random scheduler (including the case of a random sequential scheduler) and arbitrary edge dynamics to be a polynomial in the number of nodes, rather than exponential, ensuring
an overall multiplicative polylogarithmic memory overhead
for the synchronizer. Basically, each node  remembers its opinion at each phase $i$ (and there are $O(\log^2 n)$ such phases); if a node $u$ in phase $i$ pulls from a (subset of) neighbors $v$ during a stage, collecting phase $i$ opinions until it gets $k$ such opinions; at this point, it adopts the minority of the collected opinions and moves to phase $i+1$. 

In future work, we plan to investigate whether we can extend the $\delta$-Synchronizer also to work in asynchronous environments, determine its overhead in terms of runtime and bits exchanged, and investigate whether it can work under certain classes of  (non-stabilizing) network dynamics (e.g., edge recurrent, snapshot connected, etc.).

\section{Conclusion and Future Work}\label{sec:conclude}
We presented the first deterministic synchronizer for non-synchronous anonymous dynamic networks that makes no assumptions about edge dynamics or eventual network stabilization, distinguishing our work from the prior synchronizer literature. 
In future work, we plan to  implement and test our synchronizer on the applications outlined in Section~\ref{sec:appl} and others that we may find suitable. One major drawback of our synchronizer in that it requires that the underlying synchronous algorithm $\synchAlgo$ correctly work under arbitrary edge--dynamics, restricting the application of the synchronizer to several synchronous algorithms that only work under more constrained edge--dynamics (e.g., edge-recurrent, random dynamics, or the connected snapshot model). We plan to investigate whether our synchronizer, under a random scheduler, would be able to approximate edge-recurrent or random edge--dynamics synchronous environments. It is worth noting that our \synchronizer-Synchronizer is trivially able to simulate any degree-bound restrictions on the synchronous dynamic network, since 
$\synchAlgo$ assumes a  maximum node degree bound of $\Delta$ for  $\synchGraph$ if and only if our transformation  holds for  \synchronizer($\synchAlgo)$ under a $\NonsynchGraph$ restricted to having maximum degree $\Delta$ (this claim is also true for non-uniform node degree bounds, with a small modification of the algorithm).

Another indirect contribution of this work is that  we show that focusing on the ``persistent'' edges for synchronization in highly dynamic scenarios is powerful enough to guarantee the equivalence of the synchronous and semi-synchronous models in our transformation, a concept that had already been successfully explored in the context of local mutual exclusion in ~\cite{daymude2022-mutex}. 
Lastly, we plan to investigate whether we can adapt the transformation from semi-synchronous to asynchronous dynamic network environments of~\cite{daymude2022-mutex}, to extend the results in this paper to asynchronous scenarios.

\bibliographystyle{plain}
\bibliography{ref}

\appendix

\end{document}